\documentclass[11pt]{article} % use larger type; default would be 10pt
\usepackage[margin=2.3cm]{geometry}
\usepackage{amsthm}
\usepackage{amsmath}

\usepackage{graphicx}
\usepackage{subcaption}
\usepackage{authblk}

\newtheorem{theorem}{Theorem}[section]
\newtheorem{lemma}[theorem]{Lemma}

\newtheorem{corollary}[theorem]{Corollary}

\newenvironment{definition}[1][Definition]{\begin{trivlist}
\item[\hskip \labelsep {\bfseries #1}]}{\end{trivlist}}
\newenvironment{example}[1][Example]{\begin{trivlist}
\item[\hskip \labelsep {\bfseries #1}]}{\end{trivlist}}
\newenvironment{remark}[1][Remark]{\begin{trivlist}
\item[\hskip \labelsep {\bfseries #1}]}{\end{trivlist}}

\usepackage[mathscr]{eucal}

\usepackage{mathtools}

\usepackage{amsfonts}

\providecommand{\keywords}[1]
{
  \small    
  \textbf{{Keywords: }} #1
}

\title{Derivatives pricing using signature payoffs\footnote{Opinions expressed in this paper are those of the authors, and do not necessarily reflect the view of JP Morgan.}}
\author[1, 2]{Imanol Perez Arribas}
\affil[1]{J.P. Morgan, London}
\affil[2]{Mathematical Institute, University of Oxford}

\begin{document}
\maketitle

\begin{abstract}
We introduce signature payoffs, a family of path-dependent derivatives that are given in terms of the signature of the price path of the underlying asset. We show that these derivatives are dense in the space of continuous payoffs, a result that is exploited to quickly price arbitrary continuous payoffs. This approach to pricing derivatives is then tested with European options, American options, Asian options, lookback options and variance swaps. As we show, signature payoffs can be used to price these derivatives with very high accuracy.
\end{abstract}


\section{Introduction}

The idea of representing functions as linear functionals on some basis is not new. For instance, a real-valued function in $\mathbb{R}^d$ is continuous if and only if it can be locally well approximated by polynomials (\cite{weierstrass}). Hence, we could study such continuous functions by analysing their linear effects on some polynomial basis on $\mathbb{R}^d$.

The signature of a path is a mapping that, in a way, plays an analogous role to a polynomial basis but on path space. A path is a continuous map defined on an interval $[0, T]$ and taking values on a $d$-dimensional Euclidean space $\mathbb{R}^d$.  Its signature is a sequence of iterated integrals \cite{lyonsbook, learningpast} (a precise definition is made in Definition \ref{def:signature}). It turns out that continuous functions on signatures can be locally well approximated by linear functions on signatures.

This offers a wide range of applications in various fields, including finance. Suppose we have some (possibly high-dimensional) price paths, and we want to study the effects of an unknown real-valued continuous function on price paths. If all we know about the function is its image on a finite set of such price paths, we could infer more knowledge about it by taking a look to its linear effects on the signature of the price paths. We can then use this linear representation of the unkown function in terms of the signature to make predictions about the value of the function on unseen price paths.

The objective of this paper is to leverage this linear representation of functions on paths to quickly price arbitrary financial derivatives. We do so by introducing in Section \ref{sec:signature payoffs} a family of derivatives called \textit{signature payoffs}. Similarly to polynomials, these signature payoffs form an algebra: the sum of two signature payoff is a signature payoff; the product of two signature payoffs is a signature payoff. Moreover, signature payoffs are fast to price and can approximate continuous payoffs -- i.e. continuous real-valued functions defined on price paths -- arbitrarily well.

In Section \ref{sec:experiments} we applied the methodology to approximate prices of European options, American options, Asian options, lookback options and variance swaps using signature payoffs, and the signature approach proved to be very accurate at pricing them.

%Hence, the sequence of iterated integrals from which the signature is constructed forms a natural linear \textit{basis} which has been successfully exploited in machine learning \cite{ML1, ML2, ML3, ML4}.

\section{Notation}

Given a $d$-dimensional Euclidean space $\mathbb{R}^d$, we define the tensor algebra $$T((\mathbb{R}^d)) := \{ (a_i)_{i\geq 0} : a_i \in (\mathbb{R}^d)^{\otimes i}\}.$$ As its name suggests, $T((\mathbb{R}^d))$ is an algebra with the sum $+$ and tensor product $\otimes$. Similarly, define $$T(\mathbb{R}^d) := \{ (a_i)_{i\geq 0} : a_i \in (\mathbb{R}^d)^{\otimes i}\mbox{ and }\exists N\geq 0\mbox{ such that }a_i=0 \mbox{ } \forall i\geq N\},$$ which is a subalgebra of $T((\mathbb{R}^d))$. We will also consider the truncated tensor algebra of order $n\in \mathbb{N}$, defined as $T^{n}(\mathbb{R}^d) := \{(a_i)_{i=0}^n : a_i \in (\mathbb{R}^d)^{\otimes i}\}$.

Notice that there exists a natural inclusion $T((\mathbb{R}^d)^\ast) \rightarrow T((\mathbb{R}^d))^\ast$ (\cite{lyonsbook}), where $(\mathbb{R}^d)^\ast$ and $T((\mathbb{R}^d))^\ast$ denote the dual space of $\mathbb{R}^d$ and $T((\mathbb{R}^d))$ respectively, i.e. the space of all linear functionals.

\section{The signature of a path}\label{sec:signature}

We shall begin by defining the space of $d$-dimensional paths with bounded $p$-variation.

\begin{definition}
Let $p\geq 1$, and let $X\in C([0, T]; \mathbb{R}^d)$. The $p$-variation of $X$ is defined as

$$\lVert X \rVert_p := \left (\sup_{\{t_i\}_i\subset [0, T]} \sum_i |X_{t_{i+1}} - X_{t_i}|^p \right )^{1/p},$$ where the suppremum is taken over all partitions of $[0, T]$. The space of continuous paths of bounded $p$-variation is then defined as

$$\mathcal{V}^p([0, T]; \mathbb{R}^d) := \{X\in C([0, T]; \mathbb{R}^d) : \lVert X \rVert_p < \infty\},$$ which is a Banach space with the norm $\lVert X \rVert_{\mathcal{V}^p([0, T]; \mathbb{R}^d)} := \lVert X \rVert_p + \max_{0\leq t \leq T} |X_t|$.
\end{definition}

We may now define the signature of a $d$-dimensional path $X$, which is an element of the tensor algebra $T((\mathbb{R}^d))$.

\begin{definition}\label{def:signature}
Let $p\geq 1$. The signature of a path $X\in \mathcal{V}^p([0, T]; \mathbb{R}^d)$ is defined as

$$S(X) := (1, X^1, X^2, \ldots) \in T((\mathbb{R}^d))$$ where

$$X^n := \underset{0<u_1<\ldots<u_n<T}{\int\ldots\int} dX_{u_1}\otimes \ldots \otimes dX_{u_n}\in \left (\mathbb{R}^d \right )^{\otimes n},$$ provided the iterated integrals are well defined. Similarly, the truncated signature of order $n$ is defined as $S^n(X) := (1, X^1, X^2, \ldots, X^n) \in T^{n}(\mathbb{R}^d)$.
\end{definition}

\begin{remark}
The nature of the path $X$ will determine how the iterated integrals are defined. For example, if $X\in \mathcal{V}^p([0, T]; \mathbb{R}^d)$ with $1\leq p < 2$ the integrals can be defined in the sense of Young (\cite{lyonsbook}). If $X$ is a sample path of a Brownian motion, on the other hand, the integrals can be understood in the sense of It\^o or Stratonovich (\cite{frizvictoir}).
\end{remark}

\begin{example}
Let $X$ be a path for which the signature is well defined. The first term of the signature is given by

$$X^1 = \int_0^T dX_{u} = X_T - X_0 \in \mathbb{R}^d.$$

In other words, the first term of the signature is equal to increment of the path over the period $[0, T]$. The tensor $X^2$, on the other hand, is just the Levy area of the path: the signed area between the path and the chord that joins the initial point $X_0$ and the ending point $X_T$. Higher order terms of the signature capture other aspects of the path, but their geometric intuition becomes less clear.
\end{example}

Signatures are unique representations of paths up to tree-like equivalences (\cite{uniqueness, uniquenessrough}) that are parametrisation invariant (\cite[Lemma 2.12]{learningpast}). The reduction of the infinite dimensional group of reparametrisations makes signatures a concise way of representating paths, so that it makes sense to try to describe functions on paths by taking a look to their effects on path signatures. This is the idea we followed in Section \ref{sec:signature payoffs}, where we define signature payoffs as payoffs on price paths that act through their signatures in a linear way. As we will see, this class of payoffs is rich enough to approximate continuous payoffs to arbitrary accuracy.

Linear functionals on signatures will be key objects in this paper. Therefore, we will introduce some notation that will somewhat simplify their description. A word of length $n$ and alphabet $\{1, 2, \ldots, d\}$ is a sequence $I=(i_1, i_2, \ldots, i_n)\subset \{1, 2, \ldots, d\}^n$. The empty word, which is the only word of length zero, will be denoted by $()$. Given a word $I=(i_1, i_2, \ldots, i_n)$ and a $d$-dimensional path $X$ for which the signature is well defined, we define the projection $\pi^I\in T((\mathbb{R}^d)^\ast)$ as $\pi^I(S(X)) := \underset{0<u_1<\ldots<u_n<T}{\int\ldots\int} dX_{u_1}^{i_1}\ldots  dX_{u_n}^{i_n}$.

\begin{definition}
Let $(\Omega, \mathbb{P}, \mathcal{F})$ be a probability space, and let $X$ be a $d$-dimensional stochastic process. Assume that the signature of $X$ is well defined almost surely, and that $\mathbb{E}[S(X)]$ is finite. Then, $\mathbb{E}[S(X)]$ will be called the expected signature of the stochastic process $X$.
\end{definition}

Similarly to the case of the polynomial map discussed in the introduction, the expected signature of a stochastic process turns out to contain enough information about the law of the process to completely determine it (\cite{fawcett, ESlaw}).

Before defining signature payoffs, we will introduce an \textit{augmention} of a one-dimensional path.

\begin{definition}\label{def:augmented}
Let $X\in \mathcal{V}^p([0, T]; \mathbb{R})$ be a continuous path of bounded $p$-variation, with $p\geq 1$. The \textit{augmention} of $X$ is the path $\widehat{X}\in \mathcal{V}^p([0, T];  \mathbb{R}^{3})$ defined as $\widehat{X}_t :=\left (t, X_t, \frac{X_0}{T}t\right )$. The space of all augmented continuous paths with bounded $p$-variation will be denoted by $\mathcal{A}^p([0, T])$:

$$\mathcal{A}^p([0, T]) := \{\widehat{X}:X\in \mathcal{V}^p([0, T]; \mathbb{R})\}\subset \mathcal{V}^p([0, T]; \mathbb{R}^3).$$
\end{definition}

Therefore, the augmention of a path is a $3$-dimensional path. A \textit{payoff function} will then be defined as a real-valued mapping on augmented paths.

\begin{definition}
Let $p\geq 1$. A payoff function $\mathscr{P}$ is a mapping $\mathscr{P}:\mathcal{A}^p([0, T]) \rightarrow \mathbb{R}$.
\end{definition}

\begin{example}
\begin{enumerate}
\item European payoffs are payoffs of the form $\mathscr{P}(\widehat{X}) := \psi(X_T)$ for all $\widehat{X}\in \mathcal{A}^p([0, T])$, for some continuous $\psi:\mathbb{R}\rightarrow \mathbb{R}$.

\item The payoff of a lookback call option with floating strike is defined as $\mathscr{P}(\widehat{X}) := \max_{0\leq t \leq T} X_t - X_T$ for all $\widehat{X}\in \mathcal{A}^p([0, T])$.
\end{enumerate}

Notice that in these two examples, the payoff function is continuous. This is true for many other derivatives such as American options, Asian options, variance swaps, etc.
\end{example}

\section{Signature payoffs}\label{sec:signature payoffs}

We will now define \textit{signature payoffs}.

\begin{definition}
Let $\ell \in T((\mathbb{R}^3)^\ast)$. We define the $\ell$-signature payoff $\mathscr{S}^{\ell}$ of maturity $T>0$ as the payoff that, given a price path $\widehat{X}\in \mathcal{A}^p([0, T])$ for which the signature is well defined, pays to the holder of the derivative an amount of $\mathscr{S}^{\ell}(\widehat{X}) :=\ell(S(\widehat{X}))$ at maturity. In other words, $\mathscr{S}^{\ell}$ is the mapping defined as

\begin{align}
\mathscr{S}^{\ell}:\mathcal{A}^p([0, T])&\rightarrow \mathbb{R}\\
\widehat{X}&\mapsto \ell(S(\widehat{X})).
\end{align}
\end{definition}

\begin{example}

\begin{enumerate}
\item Let $K>0$, and set $\ell_K(S(\widehat{X})) := (-K \pi^{()} + \pi^{(2)}+ \pi^{(3)})(S(\widehat{X})) = -K + (X_T - X_0) + (X_0 - 0) = X_T - K$. Then, the $\ell_K$-signature payoff $\mathscr{S}^{\ell_K}$ pays to the holder of the derivative an amount of $X_T - K$ at time $T$. Therefore, $\mathscr{S}^{\ell_K}$ corresponds to a forward contract with delivery price $K$.

\item Let $K>0$. Define $\ell_K$ as $\ell_K(S(\widehat{X})) :=\left (-K\pi^{()} + \frac{1}{T}\pi^{(2,1)}+\pi^{(3)}\right )(S(\widehat{X})) = -K + \big ( \frac{1}{T}\int_0^T X_sds - X_0 \big ) + (X_0 - 0) = \frac{1}{T}\int_0^T X_sds - K$. The payoff $\mathscr{S}^{\ell_K}$ is then an Asian forward.

\end{enumerate}
\end{example}

As we see, some simple payoffs can actually be seen as signature payoffs, for some appropriate $\ell\in T((\mathbb{R}^3)^\ast)$. The natural question would then be the following: how rich is the family of signature payoffs? What other payoffs can they approximate? We will now show that signature payoffs approximate continuous payoffs arbitrarily well.

\begin{lemma}
The linear forms on $T((\mathbb{R}^d))$ induced by $T((\mathbb{R}^d)^\ast)$, when restricted to the range $S(\mathcal{V}^p([0, T]; \mathbb{R}^d))$ of the signature, form an algebra of real-valued functions.
\end{lemma}

\begin{proof}
See \cite[Theorem 2.15]{lyonsbook}.
\end{proof}

\begin{theorem}\label{th:approximation}
Let $p\geq 1$, and let $\mathscr{P}:\mathcal{A}^p([0, T])\rightarrow \mathbb{R}$ be a payoff function. Let $E\subset \mathcal{A}^p([0, T])$ be a compact set. Assume that for each $\widehat{X}\in E$, the signature $S(\widehat{X})$ is a $p$-geometric rough path (\cite[Definition 3.13]{lyonsbook}).  Then, assuming that $\mathscr{P}$ is continuous, given any $\varepsilon>0$ there exists a linear functional $\ell_\varepsilon\in T((\mathbb{R}^3)^\ast)$ such that

$$\left |\mathscr{P}(\widehat{X}) - \mathscr{S}^{\ell_\varepsilon}(\widehat{X})\right |<\varepsilon \quad \forall \widehat{X}\in E.$$
\end{theorem}

\begin{proof}
Let $\widehat{X}\in \mathcal{A}^p([0, T])$. Since at least one coordinate of $\widehat{X}$ is monotone, its signature $S(\widehat{X})$ determines $\widehat{X}$ uniquely (\cite{uniqueness, uniquenessrough}). Therefore, the mapping $\mathscr{P}:\mathcal{A}^p([0, T])\rightarrow \mathbb{R}$ induces a mapping $\widehat{\mathscr{P}}:S(\mathcal{A}^p([0, T]))\rightarrow \mathbb{R}$ such that $\mathscr{P}(\widehat{X}) = \widehat{\mathscr{P}}(S(\widehat{X}))$, where $S(\mathcal{A}^p([0, T]))$ is equipped with the topology induced by the signature mapping. Since $\mathscr{P}$ is continuous, the mapping $\widehat{\mathscr{P}}$ is also continuous.

Let $\varepsilon>0$. It is clear that linear forms on $T((\mathbb{R}^3))$ induced by $T((\mathbb{R}^3)^\ast)$ separate points and contain constant functions. Since by the previous lemma it is also an algebra, it follows from the Stone--Weierstrass Theorem that there exists an $\ell_\varepsilon\in T((\mathbb{R}^3)^\ast)$ such that $|\widehat{\mathscr{P}}(a) - \ell_\varepsilon(a)|<\varepsilon$ for all $a\in S(E)$. Thus,

$$ \left |\mathscr{P}(\widehat{X}) - \mathscr{S}^{\ell_\varepsilon}(\widehat{X})\right |<\varepsilon \quad\forall \widehat{X}\in E.$$
\end{proof}

\begin{remark}
The assumption that $S(\widehat{X})$ is a $p$-geometric rough path is used to ensure uniqueness of the signature, as per \cite{uniquenessrough}. If $X$ has bounded variation (i.e. $p=1$) then the signature, defined as Riemann--Stieltjes integrals, will always be a geometric rough path. If $X$ is a semimartingale, we need to define $S(\widehat{X})$ using Stratonovich integrals to make sure that the signature of $\widehat{X}$ is a geometric rough path.
\end{remark}

\begin{example}
It is important to have examples of payoffs that satisfy the conditions of Theorem \ref{th:approximation}. Most financial derivatives, both vanilla and exotic options, are included in the framework of Theorem \ref{th:approximation}. Some examples include European options, American options, Asian options, lookback options and variance swaps. Notice that barrier options, for example, are not included in this setting in principle because the discontinuity on the barrier implies that they do not satisfy the continuity hypothesis of Theorem \ref{th:approximation}. However, one could fix this by \textit{smoothening} barrier payoffs -- for example, with barrier bending.
\end{example}

A major application of Theorem \ref{th:approximation} is that as we have found a family of derivatives that approximates continuous payoffs, we could price a given derivative or a basket of derivatives by approximation using signature payoffs. As we will see in Section \ref{sec:pricing}, signature payoffs can be quickly priced. Therefore, pricing a derivative by approximation with signatures becomes especially interesting when pricing the derivative is computationally expensive. In Section \ref{sec:experiments} we will see that this approach to pricing arbitrary continuous payoffs is extremely accurate.

\section{Pricing signature payoffs}\label{sec:pricing}

As we have seen in Section \ref{sec:signature payoffs}, signatures can approximate to arbitrary accuracy continuous payoffs. Therefore, as mentioned at the end of that section, we could exploit this to price these derivatives: given a continuous payoff $\mathscr{P}$, one would find an $\ell$-signature payoff that approximates it to the desired accuracy. Then, the fair value of the derivative with payoff $\mathscr{P}$ could be approximated by the fair value of the $\ell$-signature. However, in order for this approach to pricing derivatives to be useful, we have to show that signature payoffs can be priced fast. In that case, one would be able to price fast derivatives that are computationally expensive to price.

Let $\mathscr{S}^\ell$ be a signature payoff. By the Fundamental Theorem of Asset Pricing (\cite{fundamental theorem}), under the assumption of no free lunch with vanishing risk (see \cite[Definition 2.8 (ii)]{fundamental theorem}) the fair value of $\mathscr{S}^\ell$ will be given by $Z_T \mathbb{E}^\mathbb{Q}[\mathscr{S}^\ell(\widehat{X})]$, with $Z_T$ the discount factor for the interval $[0, T]$ and $\mathbb{Q}$ the risk-neutral measure.

However, since the signature payoffs were defined to be linear functions on signatures, we have:

\begin{align*}
Z_T \mathbb{E}^\mathbb{Q}[\mathscr{S}^\ell(\widehat{X})] = Z_T \mathbb{E}^\mathbb{Q}[\ell(S(\widehat{X}))] = Z_T \ell \left ( \mathbb{E}^\mathbb{Q}[S(\widehat{X})] \right ).
\end{align*}

The expected signature of $\widehat{X}$ under the risk-neutral measure $\mathbb{Q}$ only depends on the risk-neutral measure (i.e. the market) which is independent of the signature payoff.

Suppose now that one has a basket of derivatives $\{\mathscr{P}_1, \ldots, \mathscr{P}_N\}$. These payoffs could be expensive to price, so that pricing the entire basket would take a considerable amount of time. Using Theorem \ref{th:approximation}, one could use the following procedure to price the basket:

\begin{enumerate}
\item Replace the basket $\{\mathscr{P}_1, \ldots, \mathscr{P}_N\}$ with a basket of signature payoffs $\{\mathscr{S}^{\ell_1}, \ldots, \mathscr{S}^{\ell_n}\}$, where the encodings $\ell_i$ were precoumputed.
\item Using current market conditions, compute the expected signature of the augmented price path under the risk-neutral measure $\mathbb{E}^\mathbb{Q}[S(\widehat{X})]$.
\item Evaluate each linear functional $\ell_i$ on the computed expected signature in order to get the fair values of $\{\mathscr{S}^{\ell_1}, \ldots, \mathscr{S}^{\ell_n}\}$, which are our approximations of the fair values of $\{\mathscr{P}_1, \ldots, \mathscr{P}_N\}$.
\end{enumerate}

In the above procedure, step 1 takes no time, as the encodings $\ell_i$ can be precomputed because they don't depend on market conditions. Step 3 isn't expensive either, since essentially the evaluation of each $\ell_i$ only consists on computing the inner product of two vectors. In general, step 2 is the only step that could take some time. However, since the expected signature is computed once and it is then used for the entire basket, the procedure can be significantly faster than pricing each derivative individually.

Moreover, in many situations the expected signature $\mathbb{E}^\mathbb{Q}[S(\widehat{X})]$ can actually be computed explicitly if we assume some model on the price paths, as we will see in Section \ref{sec:bs}, making its computation inexpensive. This would imply that one could very quickly price single derivatives or even entire baskets of derivatives, following the three steps above.

\subsection{Pricing signature payoffs under the Black--Scholes model}\label{sec:bs}

As we have seen, the problem of pricing any $\ell$-signature payoff is reduced to the problem of computing $\mathbb{E}^\mathbb{Q}[S(\widehat{X})]$, under the risk-free measure. In this section we will study how to compute the expected signature when we assume that the underlying stock price follows the Black--Scholes model.

Under the Black--Scholes model the $\mathbb{Q}$-dynamics of the price path is given by

\begin{equation}\label{eq:dynamics}
dX_t = rX_tdt + \sigma X_tdW_t,
\end{equation} with $r$ and $\sigma$ the interest rate and volatility respectively, which are assumed to be constant, and $W$ a standard Brownian motion. The augmented path $\widehat{X}_t = \left (t, X_t, \frac{t}{T}X_0\right )$ will then be a diffusion process satisfying the SDE

$$d\widehat{X} = \mu(\widehat{X}_t) dt + V(\widehat{X}_t)\cdot dW_t$$ where $W$ is a standard Brownian motion and $\mu:\mathbb{R}^3\rightarrow \mathbb{R}^3$, $V:\mathbb{R}^3\rightarrow \mathbb{R}^3$ are given by

$$\mu(x_1, x_2, x_3) := \left (1, \frac{X_0}{T}, rx_3\right ),\quad V(x_1, x_2, x_3) = (0, 0, \sigma x_3)$$ and the initial condition is $\widehat{X}_0 = (0, 0, X_0)$.

In this section, we will define $\Phi(t, x) := \mathbb{E}^\mathbb{Q}[S(\widehat{X}_{[0, t]}) | X_0=x]$ the expected signature under the risk-neutral measure of the augmented path $\widehat{X}$ over the interval $[0, t]$, conditioned on the initial value, where the signature is understood in the sense of Stratonovich. Our goal is to find $\Phi(T, X_0)$, with $X_0$ the spot price and $T>0$ the time to maturity. Then, the fair value of an $\ell$-signature payoff would be $Z_T \ell(\Phi(T, X_0))$.

It turns out that the function $\Phi$ satisfies a parabolic partial differential equation, as shown in \cite[Theorem 4.7]{hao}.

\begin{theorem}\label{th:hao}
Let $X$ be a $d$-dimensional It\^o diffusion process

$$dX_t = \mu(X_t)dt + V(X_t)\cdot dW_t$$ with $W$ a $k$-dimensional Brownian motion and

$$\mu(x) = (\mu^{(i)}(x))_{1\leq i \leq d},\quad V(x)=(V_i^{(j)}(x))_{1\leq i \leq k, 1\leq j \leq d}.$$

Assume that $\mu$ and $V$ satisfy the globally Lipschitz condition. Let $A$ be the infinitesimal generator of the process $X$, and let $\Phi_n(t, x)$ denote the $n$\textsuperscript{th} term of $\Phi(t, x)=\mathbb{E}^x[S(X_{[0, t]})] \in T((\mathbb{R}^d))$. Then, $\Phi_n:[0, T]\times \mathbb{R}^d\rightarrow  (\mathbb{R}^d)^{\otimes n}$ satisfies the following parabolic PDE, for every $n\geq 2$:

\begin{align*}
\left (-\partial_t + A\right)\Phi_n(t, x) &= - \left (\sum_{j=1}^d \mu^{(j)}(x)e_j\right )\otimes \Phi_{n-1}(t, x) \\
&-\sum_{j=1}^d \left ( \sum_{j_1=1}^dV^{(j_1)}(x)^T V^{(j)}(x)e_{j_1}\right )\otimes \partial _{x_j}\Phi_{n-1}(t, x)\\
&-\left (\dfrac{1}{2}\sum_{j_1, j_2=1}^d V^{(j_1)}(x)^T V^{(j_2)}(x)(x)e_{j_1}\otimes e_{j_2}\right )\otimes \Phi_{n-2}(t, x),
\end{align*} and $\Phi_1(t, x)$ satisfies

\begin{align*}
\left (-\partial_t + A\right)\Phi_1(t, x) &= - \left (\sum_{j=1}^d \mu^{(j)}(x)e_j\right ).
\end{align*}

Moreover, $\Phi_n$ satisfies the initial condition $\Phi_n(0, \cdot) = 0$ for $n\geq 1$, and $\Phi_0 \equiv 1$.

\end{theorem}

Using the theorem above, we may now state and prove a theorem that gives an explicit recurrent relation for the expected signature of the Black--Scholes model.

\begin{theorem}\label{th:ES}
Let $X$ follow the dynamics given by \eqref{eq:dynamics}, with constant interest rate $r$ and volatility $\sigma$. Consider its augmented path $\widehat{X}\in \mathcal{A}^p([0, T])$, and define the expected signature up to time $t$ and starting value $x$, $\Phi(t, x) := \mathbb{E}[S(\widehat{X}_{[0, t]})|X_0=x]$.

Define $E_n(t) := \exp\left (\left (nr + n(n-1)\frac{\sigma^2}{2}\right) t \right )$ for $0\leq t \leq T$, and given a word $I$ with alphabet $\{1, 2, 3\}$ let $\alpha(I)$ denote the number of appearences of the letters 2 or 3 in $I$.

Then, there exists a function $F:[0, T]\rightarrow T((\mathbb{R}^3))$ such that for every word $I$ of length at least 1, we have
%Then, for every $n\geq 1$, there exists a function $F_n:[0, T]\rightarrow (\mathbb{R}^3)^{\otimes n}$ such that, for every word $I$ of length $n$,

\begin{equation}\label{eq:projection}
\pi^I(\Phi(t, X_0)) = X_0^{\alpha(I)} \pi^I(F(t)).
\end{equation}

Moreover, the first order term of $F$ is given by $F_1(t) = te_1 + X_0(\exp(rt) - 1)e_2 + X_0 t e_3 / T$, and for $n\geq 2$
the n\textsuperscript{th} term of $F$ is characterised by the following:

\begin{enumerate}
\item $\pi^{(1, i_2, \ldots, i_n)}(F(t)) = \int_0^t E_{\alpha((1, i_2, \ldots, i_n))}(s) \pi^{(i_2, \ldots, i_n)}(F(t - s))ds$.
\item $\pi^{(2, 1, i_3, \ldots, i_n)}(F(t)) = (r + \sigma^2\alpha((1, i_3, \ldots, i_n))) \int_0^t E_{\alpha((2, 1, i_3, \ldots, i_n))}(s) \pi^{(1, i_3, \ldots, i_n)}(F(t-s))ds$.
\item $\pi^{(2, 2, i_3, \ldots, i_n)}(F(t)) = (r + \sigma^2\alpha((2, i_3, \ldots, i_n)))\int_0^t E_{\alpha((2, 2, i_3, \ldots, i_n))}\pi^{(2, i_3, \ldots, i_n)}(F(t-s))ds$\newline  $+ \frac{1}{2}\sigma^2 \int_0^t E_{\alpha((2, 2, i_3, \ldots, i_n))}(s) \pi^{(i_3, \ldots, i_n)}(F(t-s))ds$.
\item $\pi^{(3, i_2, \ldots, i_n)}(F(t)) = \frac{1}{T}\pi^{(1, i_2, \ldots, i_n)}(F(t))$.
\end{enumerate}

\end{theorem}

\begin{proof}

Assume $X$ follows the dynamics given in \eqref{eq:dynamics}. Then, $X$ is of the form

$$X_t = X_0 \exp\left (\left (r - \dfrac{\sigma^2}{2}\right )t + \sigma W_t\right ).$$

Therefore, $X_t$ is log-normally distributed, and the moments of $X_t$ are precisely given by $\mathbb{E}[X_t^n] = X_0^n E_n(t)$.

The process $\widehat{X}$ is a $3$-dimensional It\^o diffusion process with drift $\mu(x_1, x_2, x_3) = (1, X_0 / T, rx_3)$ and volatility $V(x_1, x_2, x_3) = (0, 0, \sigma x_3)$. Moreover, notice that the expected signature of $\widehat{X}$ with initial value $(x_1, x_2, x_3)$ is independent of $x_1$ and $x_2$. Therefore, by Theorem \ref{th:ES}, $\Phi(t, x)$ satisfies the following PDE for $n\geq 2$, for $(t, x) \in [0, T]\times \mathbb{R}$:

\begin{align*}
(-\partial_t + A) \Phi_n =-\left ((e_1 + rxe_2 + X_0e_3/T)\otimes \Phi_{n-1} + \sigma^2 x^2 e_2 \otimes \partial_{x} \Phi_{n-1}+\dfrac{1}{2}\sigma^2x^2 e_2\otimes e_2\otimes \Phi_{n-2}\right ),
\end{align*} and $\Phi_1$ satisfies

$$(-\partial_t + A)\Phi_1 = -(e_1 + rxe_2 + X_0e_3/T),$$ where $A$ is the infinitesimal generator of $\widehat{X}$. Moreover, $\Phi_0\equiv 1$ and $\Phi_n(0, \cdot) = 0$.

Set $f_n(t, x) = (e_1 + rxe_2 + X_0 e_3/T)\otimes \Phi_{n-1}(t, x) + \sigma^2 x^2 e_2 \otimes \partial_{x} \Phi_{n-1}(t, x)+\dfrac{1}{2}\sigma^2x^2 e_2\otimes e_2\otimes \Phi_{n-2}(t, x)$ for $n\geq 2$, and $f_1(t, x) = e_1 + rxe_2 + X_0e_3/T$. By Feynman-Kac formula, $\Phi_n$ will then be given by

$$\Phi_n(t, X_0) = \int_0^t \mathbb{E}[f_n(t - s, X_s)]ds.$$

For $n=1$, we have:

\begin{align*}
\Phi_1(t, X_0) &= \int_0^t \mathbb{E}[f_1(t-s, X_s)]ds = \int_0^t (e_1 + rE_1(s) e_2 + X_0 e_3/T) ds \\
&= te_1 + X_0(\exp(rt) - 1)e_2 + X_0 t e_3 / T.
\end{align*}

Therefore, for $I=(1)$, $I=(2)$ or $I=(3)$, we have $\pi^I(\Phi(t, X_0)) = X_0^{\alpha(I)} \pi^I(F(t))$ with $\pi^{(1)}(F(t)) := t$, $\pi^{(2)}(F(t)) := \exp(rt) - 1$ and $\pi^{(3)}(F(t)) := t/T$.

Assume now that $\pi^I(\Phi(t, X_0)) = X_0^{\alpha(I)} \pi^I(F(t))$ holds for all words $I$ of length strictly less than $n$. We will show that it also holds for words of length $n$, with $n\geq 2$. We will distinguish four steps.

\begin{itemize}
\item \underline{Case 1: $I=(1, i_2, \ldots, i_n)$.} In this case $\pi^{(1, i_2, \ldots, i_n)}(f_n(t, x)) = \pi^{(i_2, \ldots, i_n)}(\Phi_{n-1}(t, x))=\\x^{\alpha((i_2, \ldots, i_n))} \pi^{(i_2, \ldots, i_n)}(F(t))$, so that $\pi^{(1, i_2, \ldots, i_n)}(\Phi_n(t, X_0)) = X_0^{\alpha((1, i_2, \ldots, i_n))} \pi^{(1, i_2, \ldots, i_n)}(F(t))$ with $\pi^{(1, i_2, \ldots, i_n)}(F(t)) := \int_0^t E_{\alpha((1, i_2, \ldots, i_n))}(s)\pi^{(i_2, \ldots, i_n)}(F(t-s))ds$.

\item \underline{Case 2: $I=(2, 1, i_3, \ldots, i_n)$.} We have

\begin{align*}
\pi^{(2, 1, i_3, \ldots, i_n)}(f_n(t, x)) &= rx\pi^{(1, i_3, \ldots, i_n)}(\Phi_{n-1}(t, x)) + \sigma^2 x^2 \partial_x \pi^{(1, i_3, \ldots, i_n)}(\Phi_{n-1}(t, x)) \\
&=rx^{\alpha((2, 1, i_3, \ldots, i_n))} \pi^{(1, i_3, \ldots, i_n)}(F(t))\\
& +\sigma^2 \alpha((1, i_3, \ldots, i_n)) x^{\alpha((2, 1, i_3, \ldots, i_n))} \pi^{(1, i_3, \ldots, i_n)}(F(t))\\
&=x^{\alpha((2, 1, i_3, \ldots, i_n))} (r + \sigma^2\alpha((1, i_3, \ldots, i_n)))\pi^{(1, i_3, \ldots, i_n)}(F(t)).
\end{align*}

Hence, $\pi^{(2, 1, i_3, \ldots, i_n)}(\Phi_n(t, X_0)) = X_0^{\alpha((2, 1, i_3, \ldots, i_n))} \pi^{(2, 1, i_3, \ldots, i_n)}(F(t))$ with $$\pi^{(2, 1, i_3, \ldots, i_n)}(F(t)) := (r + \sigma^2 \alpha((1,i_3,\ldots,i_n)))\int_0^t E_{\alpha((2, 1, i_3, \ldots, i_n))}(s) \pi^{(1, i_3, \ldots, i_n)}(F(t-s))ds.$$

\item \underline{Case 3: $I=(2, 2, i_3, \ldots, i_n)$.} We have

\begin{align*}
\pi^{(2, 2, i_3, \ldots, i_n)}(f_n(t, x)) &= rx\pi^{(2, i_3, \ldots, i_n)}(\Phi_{n-1}(t, x)) + \sigma^2x^2 \partial_x \pi^{(2, i_3, \ldots, i_n)}(\Phi_{n-1}(t, x)) +\\
&\dfrac{1}{2}\sigma^2x^2\pi^{(i_3, \ldots, i_n)}(\Phi_{n-2}(t, x))\\
&=x^{\alpha((2, 2, i_3, \ldots, i_n))}(r + \sigma^2\alpha((2, i_3, \ldots, i_n))) \pi^{(2, i_3, \ldots, i_n)}(F(t))\\
&+ \dfrac{1}{2}\sigma^2 x^{\alpha((2, 2, i_3, \ldots, i_n))}\pi^{(i_3, \ldots, i_n)}(F(t))\\
&=x^{\alpha((2, 2, i_3, \ldots, i_n))} \bigg ((r + \sigma^2 \alpha((2, i_3, \ldots, i_n)) )\pi^{(2, i_3, \ldots, i_n)}(F(t))\\
& + \dfrac{1}{2}\sigma^2 \pi^{(i_3, \ldots, i_n)}(F(t))\bigg ).
\end{align*}

Thus, $\pi^{(2, 2, i_3, \ldots, i_n)}(\Phi_n(t, X_0)) = X_0^{\alpha((2, 2, i_3, \ldots, i_n))} \pi^{(2, 2, i_3, \ldots, i_n)}(F(t))$ with+ 

\begin{align*}
 \pi^{(2, 2, i_3, \ldots, i_n)}(F(t)) &:= \bigg (\mu + \sigma^2 \alpha((2, i_3, \ldots, i_n)) \int_0^t E_{\alpha((2, 2, i_3, \ldots, i_n))}(s) \pi^{(2, i_3, \ldots, i_n)}(F(t-s))ds \\
&\dfrac{1}{2}\sigma^2 \int_0^t E_{\alpha((2, 2, i_3, \ldots, i_n))}(s) \pi^{(i_3, \ldots, i_n)}(F(t-s))ds\bigg ).
\end{align*}

\item \underline{Case 4: $I=(3, i_2, \ldots, i_n)$.} From the definition of $\pi^{(3, i_2, \ldots, i_n)}(S(\widehat{X}))$, we have:

\begin{align*}
&\pi^{(3, i_2, \ldots, i_n)}(S(\widehat{X})) = \underset{0<u_1<\ldots<u_n<T}{\int\ldots\int} dX_{u_1}^{3} dX_{u_2}^{i_2} \ldots dX_{u_n}^{i_n} \\
&=\dfrac{X_0}{T}\underset{0<u_1<\ldots<u_n<T}{\int\ldots\int} dX_{u_1}^1 dX_{u_2}^{i_2} \ldots dX_{u_n}^{i_n} = \dfrac{X_0}{T} \pi^{(1, i_2, \ldots, i_n)}(S(\widehat{X})).
\end{align*}

Hence,

\begin{align*}
&\pi^{(3, i_2, \ldots, i_n)}(\Phi(t, X_0)) = \dfrac{X_0}{T}\pi^{(1, i_2, \ldots, i_n)}(\Phi(t, X_0)) = \dfrac{X_0^{1+\alpha((1, i_2, \ldots, i_n))}}{T} \pi^{(1, i_2, \ldots, i_n)}(F(t))\\
&=  \dfrac{X_0^{\alpha((3, i_2, \ldots, i_n))}}{T} \pi^{(1, i_2, \ldots, i_n)}(F(t))
\end{align*} so that $$\pi^{(3, i_2, \ldots, i_n)}(F(t)) = \dfrac{1}{T} \pi^{(1, i_2, \ldots, i_n)}(F(t)).$$

\end{itemize}

\end{proof}

Theorem \ref{th:ES} provides a way to explicitly find the expected signature. Therefore, we have now found a way of explicitly pricing under the Black--Scholes model any signature payoff. Since the terms in $F$ can be iteratively computed up to a given order in a fast way, we can price signature payoffs inexpensively:

\begin{corollary}
Let $\ell\in T((\mathbb{R}^3)^\ast)$, and let $\mathscr{S}^\ell$ be an $\ell$-signature payoff of maturity $T>0$. Then, under the Black--Scholes model with constant interest rate $r$ and volatility $\sigma$, the fair value of the signature payoff is given by $\exp(-rT) \ell(\Phi(T, X_0))$, where $X_0>0$ is the spot price and $\Phi$ is given by \eqref{eq:projection}.
\end{corollary}

\begin{example}
Consider the case of forward contracts, which as we have seen are also signature payoffs with $\ell_K :=  -K\pi^{()} + \pi^{(2)}+ \pi^{(3)}$. By the previous corollary, the fair value of forward contracts with delivery price $K$ will be given by

$$\exp(-rT) \ell_K(\Phi(T, X_0)) = \exp(-rT) (-K+X_0(\exp(rT) -1) + X_0) = X_0 - K\exp(-rT)$$ which corresponds to the well-known fair value of forward contracts.

\end{example}

In this section we studied what the fair value of a signature payoff is, under the Black--Scholes model, and we found an explict closed formula. A natural question would be whether a similar procedure can be followed for more complex models. Assuming that the price path follows some diffusion process, such as a local volatility model, one should be able to apply Theorem \ref{th:hao} again in order to get a similar explicit expresion for the fair value.

\section{Numerical experiments}\label{sec:experiments}

We implemented the proposed approach of pricing using signature payoffs to compute the fair price of different derivatives. We considered European call options with moneyness $99\%$, American put options with moneyness $99\%$, Asian options with moneyness $102\%$, lookback options and variance swaps with strike $20\%$, all with 1 year maturity. All of these derivatives satisfy the conditions of Theorem \ref{th:approximation}.

\begin{figure}
  \centering
  \subcaptionbox{European call option with moneyness $99\%$.}[.45\linewidth][c]{%
    \includegraphics[width=.45\linewidth]{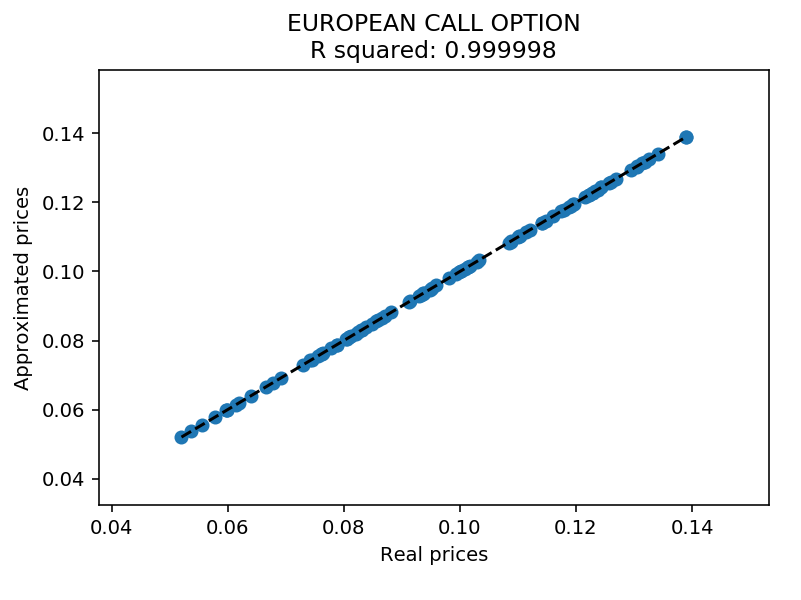}}\quad
  \subcaptionbox{American put option with moneyness $99\%$.}[.45\linewidth][c]{%
    \includegraphics[width=.45\linewidth]{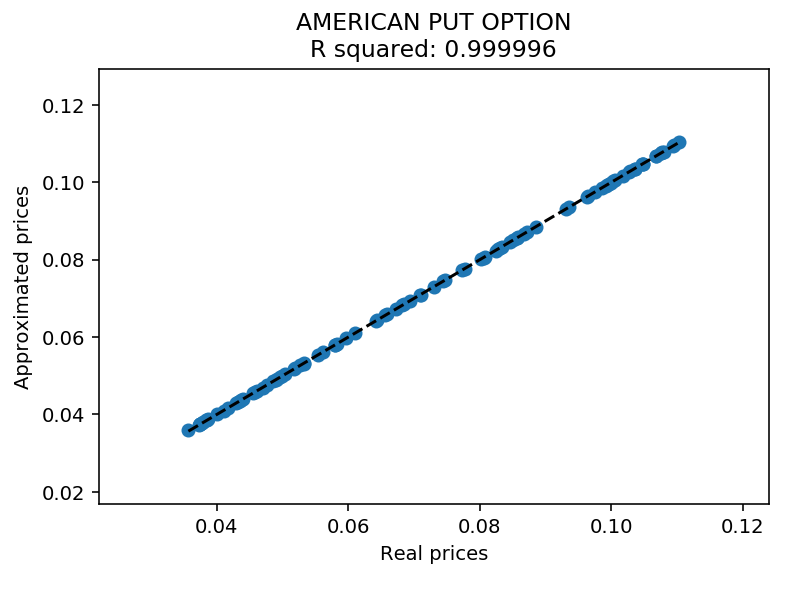}}\quad

  \bigskip

  \subcaptionbox{Asian call option with moneyness 102\%.}[.45\linewidth][c]{%
    \includegraphics[width=.45\linewidth]{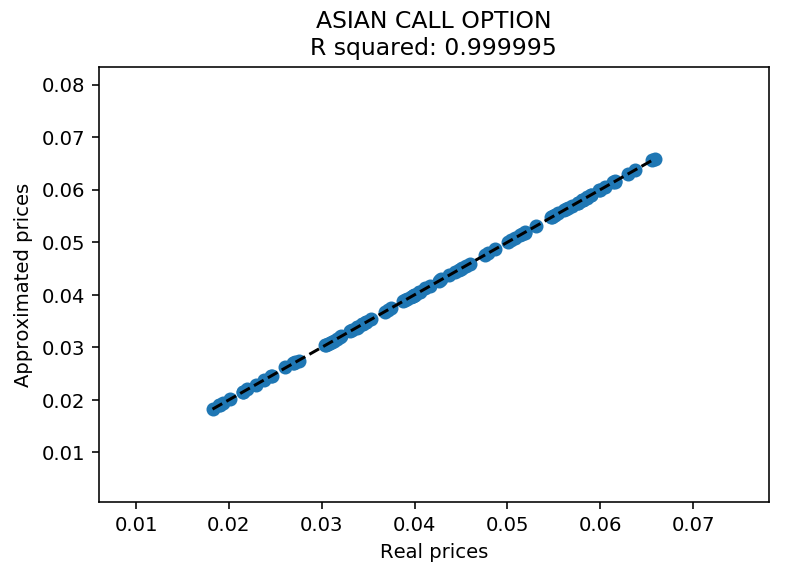}}
  \subcaptionbox{Lookback call option with floating strike.}[.45\linewidth][c]{%
    \includegraphics[width=.45\linewidth]{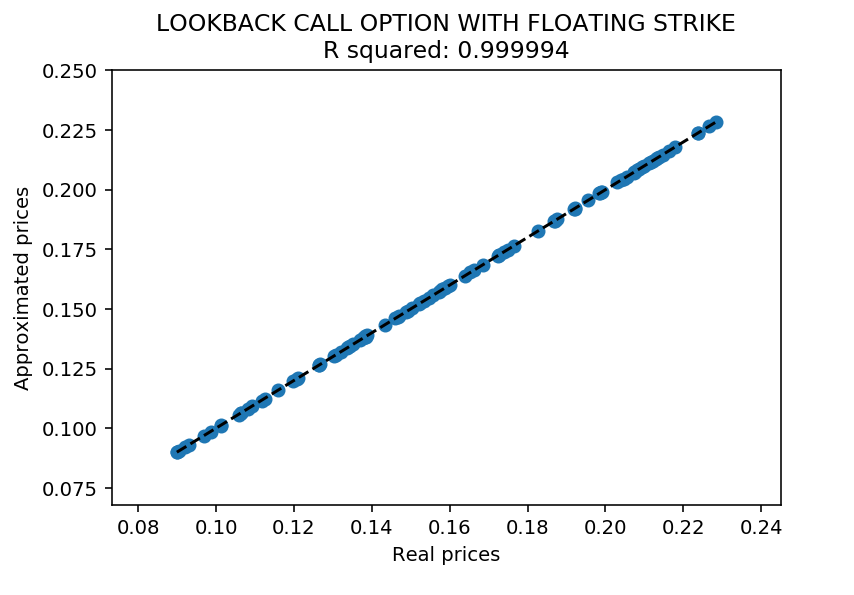}}\quad

  \bigskip

  \subcaptionbox{Variance swap with strike $20\%$.}[.5\linewidth][c]{%
    \includegraphics[width=.45\linewidth]{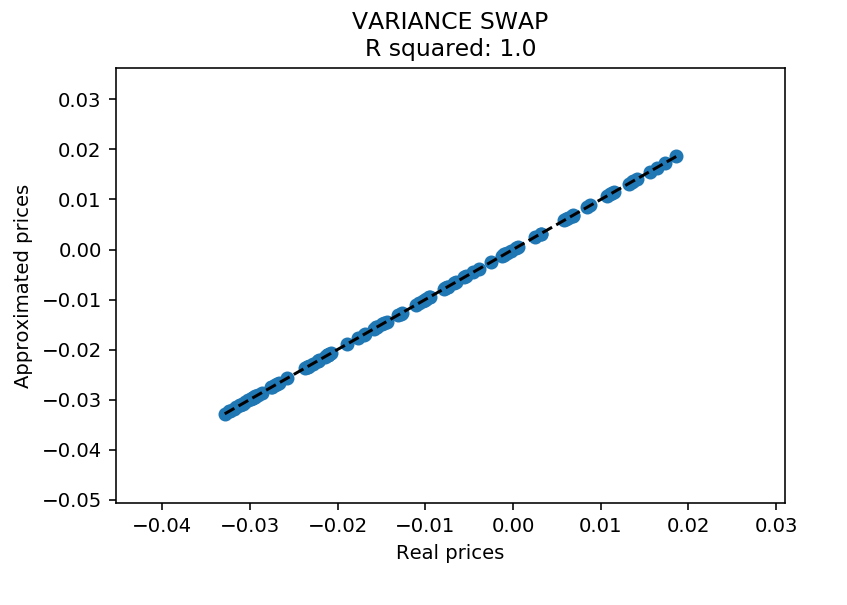}}\quad

  \caption{Numerical experiments on approximating prices of derivatives using signatures. $R^2$ was used to measure performance. In all cases, the $R^2$ was higher than 0.99999.}
\label{fig:experiment}
\end{figure}

For obvious computational reasons, we had to truncate signatures so that instead of considering linear functionals on the full signature, we considered linear functionals on the truncated signature of order $n\geq 1$. For this experiments, we fixed $n=4$, which for a 3-dimensional augmented path produces a linear functional of dimension $1 + 3 + 3^2 + 3^3 + 3^4 = 121$.

The linear functional was estimated applying linear regression against the truncated signature using a dataset of simulated market conditions following the Black--Scholes model \eqref{eq:dynamics}. A dataset of $100$ different market conditions was used for this task -- a tiny dataset for machine learning standards. We then priced the derivatives using the approximated signature payoff in an out-of-sample set of 100 market conditions, and we then compared the price we obtained with the corresponding real prices. As we see in Figure \ref{fig:experiment}, the accuracy is remarkable: we obtained an $R^2$ higher than 0.99999 in all the derivatives we considered.

\section{Conclusion}

In this paper we introduce signature payoffs, a family of derivatives that pay to the holder an amount dependent on the signature of the price process. This dependence on the signature is given by a linear functional, which makes pricing them relatively easy -- the task of computing the fair value of signature payoffs is reduced to the task of computing an expected signature. As we have seen, this expected signature is easy to compute if the price process is assumed to follow a specific model. In Section \ref{sec:bs} we studied the particular case of a Black--Scholes model, but the computations shown in that section can easily be extended to other models. Moreover, one could easily extend the framework to price signature payoffs for dividend-paying underlyings by suitably modifying the definition of the augmention of a path (Definition \ref{def:augmented}) to incorporate information about the dividends. A similar approach can be followed to price multi-asset signature payoffs.

The power of signature payoffs comes from the capability of signatures to approximate continuous functions on paths. In Theorem \ref{th:approximation} we show that signature payoffs can approximate arbitrary continuous functions to arbitrary accuracy. This makes pricing entire baskets of derivatives quick: using an accurate representation as a signature payoff of each derivative in the basket, which can be precomputed, we can use Section \ref{sec:pricing} to approximate the fair value of each derivative in the basket by the corresponding fair values of the signature payoffs. As we saw in Section \ref{sec:experiments}, this approximation turns out to be remarkably accurate: we achieved an $R^2$ higher than $0.99999$ in all the payoffs we considered.

\section{Disclaimer}

Opinions and estimates constitute our judgement as of the date of this Material, are for informational purposes only and are subject to change without notice. This Material is not the product of J.P. Morgans Research Department and therefore, has not been prepared in accordance with legal requirements to promote the independence of research, including but not limited to, the prohibition on the dealing ahead of the dissemination of investment research. This Material is not intended as research, a recommendation, advice, offer or solicitation for the purchase or sale of any financial product or service, or to be used in any way for evaluating the merits of participating in any transaction. It is not a research report and is not intended as such. Past performance is not indicative of future results. Please consult your own advisors regarding legal, tax, accounting or any other aspects including suitability implications for your particular circumstances. J.P. Morgan disclaims any responsibility or liability whatsoever for the quality, accuracy or completeness of the information herein, and for any reliance on, or use of this material in any way.

Important disclosures at: \texttt{www.jpmorgan.com/disclosures}.

\end{document}